\documentclass[letterpaper, 10pt, conference]{ieeeconf}
\IEEEoverridecommandlockouts
\usepackage{amssymb}
\usepackage{cuted}
\usepackage{amsmath,amsfonts}
\usepackage{algorithm}
\usepackage{algpseudocode}
\usepackage{cite}
\usepackage{graphicx}
\usepackage{subcaption}
\usepackage{hyperref}
\usepackage{xurl}
\usepackage{amssymb}
\usepackage{mathrsfs}
\usepackage{array}
\usepackage{mathabx}
\usepackage{xcolor}
\usepackage{mathtools}

\usepackage{amsthm}
\newtheorem{theorem}{Theorem}
\newtheorem{proposition}{Proposition}
\newtheorem{lemma}{Lemma}

\newtheorem{corollary}{Corollary}

\theoremstyle{definition}
\newtheorem{definition}{Definition}
\newtheorem{problem}{Problem}

\newcommand{\R}{\mathbb{R}}
\newcommand{\N}{\mathbb{N}}
\newcommand{\C}{\mathbb{C}}

\newcommand{\rank}{\textup{rank}}

\newcommand{\Pcal}{\mathcal{P}}
\newcommand{\Acal}{\mathcal{A}}

\newcommand{\Bcal}{\mathcal{B}}
\newcommand{\Ccal}{\mathcal{C}}
\newcommand{\Dcal}{\mathcal{D}}

\newcommand{\Scal}{\mathcal{S}}

\newcommand{\tbf}{\textbf}
\newcommand{\tA}{\tilde{A}}
\newcommand{\tB}{\tilde{B}}
\newcommand{\tC}{\tilde{C}}
\newcommand{\tD}{\tilde{D}}
\newcommand{\up}{\textup}
\newcommand{\mbf}{\mathbf}

\newcommand{\LT}{\textup{LT}}

\newcommand{\tquad}{\text{ }}

\newcommand{\StateWrap}[1]{
\State \parbox[t]{\dimexpr\linewidth-\algorithmicindent*15/10}{#1}%
}

\DeclarePairedDelimiter{\abs}{\lvert}{\rvert}
\DeclarePairedDelimiter{\abrac}{\langle}{\rangle}

\algtext*{EndIf}
\algtext*{EndFor}
\algtext*{EndWhile}

\title{\LARGE \bf Realizations of Linear Systems with Prescribed Sparsity Patterns}

\author{Yaozhi Du, Jing Shuang (Lisa) Li
\thanks{The authors are with the Department of Electrical Engineering and Computer Science, University of Michigan.  
 {\tt\small \{yzdu, jslisali\}@umich.edu}}
}
  
\begin{document}

\maketitle

\begin{abstract}

This paper solves the following problem: given a transfer function and desired sparsity patterns for its state-space realization, how can we find a realization of the function that adheres to them? 
While sparse controllers are studied in distributed control, their sparsity patterns are generally designed to match those of the plant.
In contrast, we are interested in \textit{arbitrary} sparsity patterns, which are relevant to modeling problems in sensorimotor neuroscience.
Though this problem is highly nonconvex, we solve it exactly.
We first show that the problem reduces to finding an appropriate similarity transform from the modal realization, which in turn reduces to solving a system of multivariate polynomial equations. 
We then leverage tools from algebraic geometry (Gr\"obner basis, moment method) to solve the system.
Algorithms are provided for both real- and complex-valued realization problems, and their efficacy is demonstrated on several examples.

\end{abstract}

\section{Introduction}

\label{sec:intro}

Feedback controllers are commonly used to model animal sensorimotor behaviors (e.g., reaching) \cite{diedrichsen2010coordination}.
How are such controllers (i.e., behaviors) implemented in the animal's nervous system?
The prevailing view in motor neuroscience is that the nervous system performs computations underlying control through dynamics \cite{vyas2020computation, li2025toward}.
Translated into the language of control theory, the controller is \textit{realized} on neural dynamics.
This is convenient for theorists, as it allows us to leverage standard realization techniques for modeling.
Less conveniently, biological networks are sparsely connected \cite{bullmore2012economy}, imposing sparsity constraints on the realizations and complicating the modeling problem.
The aim of this paper is to pose and solve the problem of sparse realization, enabling principled modeling of sensorimotor behavior under biological connectivity constraints.

In control theory, sparse controllers are typically studied in the context of distributed control, where sparsity translates to reduced communication and computation requirements \cite{anderson2019system, sabuau2023network, du2024distributed}.
However, such approaches are specialized to their application and do not lend themselves to our setting.
First and most importantly, distributed controllers' sparsity are designed to match that of the plant (e.g., via parameter $d$ in System Level Synthesis (SLS)-based works \cite{anderson2019system, du2024distributed}).
This is inapplicable to neuroscience modeling, where connectivity constraints are dictated by neural anatomy and may be unrelated to the plant (which itself may not even be networked).
This setting renders some existing methods intractable (e.g., SLS); a related approach \cite{li2020separating} loosens this requirement but is still unable to handle arbitrary sparsity constraints.
Second, distributed control methods \textit{design} the controller as opposed to merely realizing it.
This is again inapplicable to our setting, in which the controller can already be deduced from behavior and is therefore fixed.

In this paper, we study whether a given behavior can be implemented on a given circuit, as defined by its sparsity pattern.
We solve this exactly in the linear setting; linear controllers are widely used in models of neural feedback control \cite{diedrichsen2010coordination}, as are linear models of coarse neural dynamics \cite{nozari2024macroscopic}.
In mathematical terms: given a transfer function and a desired sparsity pattern for its state-space realization, how can we find a realization that adheres to it?
We first introduce relevant definitions (Section \ref{sec:setup}).
Then, we show that solving the sparse realization problem is equivalent to solving a system of multivariate polynomial equations (Section \ref{sec:equivalence}), and adapt tools from algebraic geometry \cite{cox1997ideals, cox1998using,lasserre2008semidefinite} to solve this problem (Section \ref{sec:grobner}). 
We conclude with simulations demonstrating the efficacy of our approach (Section \ref{sec:simulations}).

\textbf{Notation:} $\delta_{ij}$ is an $n\times n$ matrix with the $(i,j)$-th entry equal to $1$ and zeros elsewhere. For $\Bbbk = \R$ or $\C$, $\Bbbk[x_1,\cdots,x_k]\coloneq \Bbbk[x]$ is the set of multivariate polynomials in $x_1,\cdots,x_k$ with coefficients in $\Bbbk$ and $GL_n(\Bbbk)$ is the set of $n\times n$ invertible matrices with entries in $\Bbbk$. 

\section{Problem setup} \label{sec:setup}

Consider linear time-invariant (LTI) system
\begin{equation} \label{equ:abcd_realization}
    \begin{aligned}
        \dot{z}(t) &= Az(t)+By(t)\\
        u(t) &= Cz(t)+Dy(t)
    \end{aligned}
\end{equation}
with input $y(t)$ and output $u(t)$; $A\in\C^{n\times n}$, $B\in\C^{n\times m}$, $C\in\C^{q\times n}$, and $D\in\C^{q\times m}$. 
We denote it concisely as $\Sigma$:
\begin{equation}
    \Sigma = (A,B,C,D) \label{equ:ss_realization_sigma}
\end{equation}

Given rational proper transfer function matrix $\mbf{K}(s)$ with no pole-zero cancellations, $\Sigma$ is a \textbf{realization} of $\mbf{K}(s)$ if
\begin{equation} \label{eq:realization_equation}   
    \mbf{K}(s) = C(sI-A)^{-1}B + D
\end{equation}

From standard results, $\mbf{K}(s)$ has infinite realizations.
A \tbf{minimal realization} is one with the smallest possible $A$ matrix. 
In this paper, we consider only minimal realizations; accordingly, we assume $\mbf{K}(s)$ has $m$ rows and $q$ columns.

\textbf{Remark:} While we focus our discussion on continuous-time controllers, our approach applies to arbitrary causal linear systems in both continuous- and discrete-time. 



We now introduce concepts related to sparse realizations.

\begin{definition}
    A desired \textbf{support} for system $\Sigma$ is $\Scal = (\Acal,\Bcal,\Ccal,\Dcal)$ where $\Acal\in\{0,1\}^{n\times n}$, $\Bcal\in\{0,1\}^{n\times m}$, $\Ccal\in\{0,1\}^{q\times n}$, and $\Dcal\in\{0,1\}^{q\times m}$.
\end{definition}

 
 

\begin{definition}
    System $\Sigma$ \textbf{satisfies} a support $\Scal$ if
    \begin{equation}
    \begin{aligned}
        &\forall i,j\text{ s.t. } &{\Acal}_{ij} = 0,&\quad A_{ij} = 0\\
        &\forall i,j\text{ s.t. } &{\Bcal}_{ij} = 0,&\quad B_{ij} = 0\\
        &\forall i,j\text{ s.t. } &{\Ccal}_{ij} = 0,&\quad C_{ij} = 0\\
        &\forall i,j\text{ s.t. } &{\Dcal}_{ij} = 0,&\quad D_{ij} = 0
    \end{aligned}
    \end{equation}
\end{definition}


We now formalize the problems under consideration:

\begin{problem}
    Given transfer function matrix $\mbf{K}(s)$ and support $\Scal$, find $\Sigma$ that satisfies $\Scal$ and \eqref{eq:realization_equation}. \label{pb:complex_realization}
\end{problem}

\begin{problem}
    Given transfer function matrix $\mbf{K}(s)$ and support $\Scal$, find \textit{real-valued} $\Sigma$ that satisfies $\Scal$ and \eqref{eq:realization_equation}. \label{pb:real_realization}
\end{problem}

Depending on $\mbf{K}(s)$ and $\Scal$, solutions to Problems \eqref{pb:complex_realization} and \eqref{pb:real_realization} may or may not exist; this is formalized by the following:

\begin{definition}
   $\mbf{K}(s)$ is (real-) \tbf{realizable} on support $\Scal$ if there exists (real-valued) $\Sigma$ that satisfies $\Scal$ and \eqref{eq:realization_equation}.
\end{definition}



\section{Equivalence of sparse realization to multivariate polynomial equations} \label{sec:equivalence}

We now show that Problems \ref{pb:complex_realization} and \ref{pb:real_realization} can be equivalently formulated as finding solutions to systems of multivariate polynomial equations, and also derive necessary and sufficient conditions on existence of solutions.
First, recall the canonical result that we can convert between minimal realizations using similarity transforms:


\begin{lemma}[{\cite[Prop 2.3]{de2000minimal}}]
\label{thm:similarity_of_min_realizations}
    Let $\Sigma$ and $\Tilde{\Sigma} = (\Tilde{A},\Tilde{B},\Tilde{C},\Tilde{D})$, be minimal realizations of $\up{\tbf{K}}(s)$. $\exists T \in {GL}_n(\C)$ s.t.
    \begin{equation}
        \tilde{A} = T^{-1}AT, \quad \tilde{B} = T^{-1}B, \quad\tilde{C} = CT,\quad \tilde{D} = D
    \end{equation}
\end{lemma}

Lemma \ref{thm:similarity_of_min_realizations} implies that matrix $D$ is the same for all minimal realizations of $\mbf{K}(s)$. Then,


\begin{corollary} \label{corr:D}
    Given $\mbf{K}(s)$ and $\Scal$, consider $D = \lim_{s \rightarrow \infty}\mbf{K}(s)$. If $\exists i,j\quad s.t.\quad \Dcal_{ij} = 0 \quad and \quad D_{ij}\neq 0$, $\mbf{K}(s)$ is not realizable on support $\Scal$.
\end{corollary}

For the rest of the paper, we assume that Corollary \ref{corr:D} does \textit{not} hold, so that $\mbf{K}(s)$ may be realizable on support $\Scal$. 
We also assume that $\mbf{K}(s)$ admits realizations without degenerate eigenvalues in $A$; that is, it admits \textbf{modal realizations} $\Sigma_m = (A_m, B_m, C_m, D_m)$ where $A_m$ is diagonal.
We use this fact along with Lemma \ref{thm:similarity_of_min_realizations} to propose the following:



\begin{corollary}
    \label{thm:problems_equal_finding_T}
    Let $\Sigma_m$ be a minimal modal realization of $\mbf{K}(s)$, and let $\Sigma_T \coloneq (T^{-1}A_mT,T^{-1}B_m,C_mT,D_m)$ denote its transformation under some $T\in{GL}_n(\C)$.
    \begin{enumerate}
        \item A solution to Problem \ref{pb:complex_realization} exists if and only if for any minimal modal realization $\Sigma_m$ of $\mbf{K}(s)$, $\exists T\in{GL}_n(\C)$ such that $\Sigma_T$ satisfies $\Scal$.
        \item A solution to Problem \ref{pb:real_realization} exists if and only if for any minimal modal realization $\Sigma_m$ of $\mbf{K}(s)$, $\exists T\in{GL}_n(\C)$ such that $\Sigma_T$ satisfies $\Scal$ and is real-valued. In particular, if $\Sigma_m$ is real-valued, we require $T\in{GL}_n(\R)$.
    \end{enumerate}
    \label{thm:existence_of_T}
\end{corollary}

In other words, solving the problems under consideration is equivalent to finding invertible matrices $T$ to appropriately transform some minimal modal realization.
From row reduction, we see that invertible matrices are products of invertible elementary matrices corresponding to permutation of rows, row addition, or multiplication of a row by a scalar.

\begin{enumerate}
    \item Elementary permutation matrix $E^p_{ij}$: identity matrix with rows $i$ and $j$ permuted.
    \item Elementary addition matrix: identity matrix with one off-diagonal entry replaced by nonzero scalar $x$:
    \begin{equation}
        E^a_{ij}(x) = I_n+x\delta_{ij},\quad E^a_{ij}(x)^{-1} = I_n-x\delta_{ij} \label{equ:elementary_addition}
    \end{equation}
    \item Elementary multiplication matrix $E^m_{i}(x)$: identity matrix with row $i$ multiplied by nonzero scalar $x$.
\end{enumerate}

Any invertible $T$ can be reduced to the identity matrix by a sequence of left-multiplications by elementary matrices.
Without loss of generality, we apply permutations, then additions, then multiplications. 
Let $T_p$, $T_a$, and $T_m$ denote products of the relevant elementary permutation, addition, and multiplication matrices. Then,

\begin{equation}
    T_mT_aT_pT = I,\quad T^{-1} = T_mT_aT_p \label{equ:decompose_T}
\end{equation}

The order of permutations is non-unique; any order yields a valid $T_p$.
We use $\Pcal_p$ to denote the set of valid $T_p$ and further discuss in subsequent sections.
Similarly, $T_a$ is non-unique.
We adopt the following order: given $T_pT$, process columns from left to right.
Within each column $i$, cancel nonzero off-diagonal entries from top to bottom by adding appropriate scalar multiples of row $i$.
$T_m$ is also non-unique; however, it is irrelevant for our problem as row multiplications cannot change sparsity.
Thus, solutions to Problems \ref{pb:complex_realization} and \ref{pb:real_realization} exist if appropriate $T_p, T_a$ exist. First, define system $\tilde{\Sigma}$:
\begin{equation}
    \begin{aligned}
        &\tilde{\Sigma} = (\tA,\tB,\tC,D_m),\tquad \Tilde{A} = (T_aT_p)A_m(T_aT_p)^{-1}\\
        &\Tilde{B} = (T_aT_p)B_m,\tquad \tilde{C} = C_m(T_aT_p)^{-1} \label{equ:tilde_realizations}
    \end{aligned}
\end{equation}
    
\begin{theorem}
    \label{thm:realizability_and_T} $\mbf{K}(s)$ is realizable on support $\Scal$ if and only if there exists a product of permutation matrices $T_p$ and a product of addition matrices $T_a$ s.t. $\tilde{\Sigma}$ \eqref{equ:tilde_realizations} satisfies $\Scal$.
\end{theorem}

\begin{proof}
    By \eqref{equ:decompose_T} and Corollary \ref{thm:problems_equal_finding_T}, $\mbf{K}(s)$ is realizable on $\Scal$ if and only if we can find
    \begin{equation}
        \tilde{\Sigma}_a = (T_m\tA T_m^{-1}, T_m\tB,\tC T_m^{-1}, D_m)
    \end{equation}
    Since $T_m$ is diagonal and invertible, $\tilde{\Sigma}$ satisfies support $\Scal$ if and only if $\tilde{\Sigma}_a$ satisfies $\Scal$. Thus, $\mbf{K}(s)$ is realizable on $\Scal$ if and only if we can find $\tilde{\Sigma}$, which proves the claim.
\end{proof}

$T_a$ makes at most $n(n-1)$ row additions. 
By \eqref{equ:elementary_addition}, each entry of an elementary addition matrix and its inverse are a polynomial of the matrix's parameter (i.e., the scalar value $x$) with real coefficients. 
We denote those parameters as $x_1,\cdots,x_{n(n-1)}$; for compactness, we define $x\coloneq \begin{bmatrix} x_1 &\cdots &x_{n-1} 
\end{bmatrix}^\top \in \C^{n(n-1)} $.
Therefore, entries of $T_a$ and $T_a^{-1}$ are polynomials of $x$; we make this notationally explicit by writing them as $T_a(x)$ and $T_a^{-1}(x)$, where $\left[T_a(x)\right]_{ij}\in\R[x]$, $\left[T_a^{-1}(x)\right]_{ij}\in\R[x]$.
Then, we propose the following:

\begin{corollary}
    \label{thm:problem_to_poly_eqns}
    Given transfer function matrix $\mbf{K}(s)$ and support $\Scal$, let $\Sigma_m$ be a minimal modal realization of $\mbf{K}(s)$. 
    \begin{enumerate}
        \item Problem \ref{pb:complex_realization} is solvable if and only if $\exists T_p\in\Pcal_p$ and $x\in\C^{n(n-1)}$ such that $\tilde{\Sigma}$ \eqref{equ:tilde_realizations} (with $T_a(x)$) satisfies $\Scal$.
        
        \item For real-valued $\Sigma_m$, Problem \ref{pb:real_realization} is solvable if and only if $\exists T_p\in\Pcal_p$ and $x\in\R^{n(n-1)}$ such that $\tilde{\Sigma}$ \eqref{equ:tilde_realizations} (with $T_a(x)$) satisfies $\Scal$.
    \end{enumerate}
\end{corollary}

\begin{proof}
    Existence of appropriate $x$ is equivalent to existence of appropriate $T_a(x)$; claim 1) holds trivially by Theorem \ref{thm:existence_of_T}.

    For sufficiency of 2), we only need to show that $\tA$, $\tB$, $\tC$ are real-valued when a solution $\tilde{x}\in\R^{n(n-1)}$ exists. 
    As $\left[T_a(x)\right]_{ij}\in\R[x]$, $T_a(\tilde{x})\in\R^{n\times n}$. 
    As $T_p\in\R^{n\times n}$ and $\Sigma_m$ is real-valued, $\tA$, $\tB$, $\tC$, are real-valued as well. 
    For necessity of 2), Corollary \ref{thm:problems_equal_finding_T} implies that when a real-valued realization exists, there exists an invertible matrix $T$ which converts $\Sigma_m$ to a realization satisfying $\Scal$. When $\Sigma_m$ is real-valued, the proof of Lemma \ref{thm:similarity_of_min_realizations} implies that such $T\in\R^{n\times n}$. Then,  $T_pT\in\R^{n\times n}$, and when we build $T_a(x)$, all elementary addition matrices have real-valued parameters, and $\tilde{x}\in\R^{n(n-1)}$, which proves necessity.
\end{proof}

Corollary \ref{thm:problem_to_poly_eqns} implies that Problems \ref{pb:complex_realization} and \ref{pb:real_realization} reduce to solving multivariate polynomial equations. 
Since each entry of $\tA$, $\tB$ and $\tC$ is a polynomial, we write $\tA_{ij}(x)\in\R[x]$, $\tB_{ij}(x)\in\R[x]$, $\tC_{ij}(x)\in\R[x]$. 
$\tilde{\Sigma}$ satisfies $\Scal$ if and only if:
\begin{equation}
\begin{aligned}
    \tA_{ij}(x) = 0\quad\forall i,j \quad s.t. \quad \Acal_{ij} = 0 \\
    \tB_{ij}(x) = 0\quad\forall i,j \quad s.t. \quad \Bcal_{ij} = 0 \\
    \tC_{ij}(x) = 0\quad\forall i,j \quad s.t. \quad \Ccal_{ij} = 0
\end{aligned}
\label{equ:the_equivalent_multvar_poly_eqns}
\end{equation}
Thus, the feasibility of Problems \ref{pb:complex_realization} and \ref{pb:real_realization} is equivalent to the existence of solutions to polynomial system \eqref{equ:the_equivalent_multvar_poly_eqns}.

We now briefly discuss a related problem. 
While Problems \ref{pb:complex_realization} and \ref{pb:real_realization} are concerned with specific $\mbf{K}(s)$, certain supports can be used to realize \textit{many} different $\mbf{K}(s)$ of appropriate order, which is reminiscent of canonical microcircuits in the nervous system \cite{li2025toward}. 
This motivates the following definition:
\begin{definition}
    Support $\Scal$ with $\mathcal{A} \in \{0, 1\}^{n \times n}$ is \tbf{generally-realizing} if all minimal $\mbf{K}(s)$ with $n$ non-degenerate poles and appropriate input-output dimension is realizable on $\Scal$.
\end{definition}

Trivially, the support of a modal or controllable/observable canonical realization is generally-realizing.
Similarly, any support that is strictly denser than a generally-realizing support is itself generally-realizing.
For Problem \ref{pb:complex_realization}, solutions exist if the given $\Scal$ is generally-realizing.
Determining whether given $\Scal$ is generally-realizing is beyond the scope of this work; however, we leverage the above polynomial analysis to present a novel generally-realizing support.


\begin{lemma}
    For SISO systems with $n$ states, the following (with $\Dcal = 1$) is a generally-realizing support $\Scal$:
    \begin{equation}
    \begin{aligned}
        \Acal_{ij} &= 1 \text{ if $i=1$, $j=1$ or $i=j$,} & \text{$0$ otherwise} \\
        \Bcal_{ij} &= 1 \text{ if $i=1$ or $i=n$,} &\text{$0$ otherwise} \\
        \Ccal_{ij} &= 1 \text{ if $j=1$ or $j=n$,} &\text{$0$ otherwise} \\
    \end{aligned}
    \end{equation}
\end{lemma}    
\begin{proof}
    Any proper $\mbf{K}(s)$ can be written as $\mbf{K}(s) = \sum_{i=1}^n\frac{c_i}{s-p_i}+d$, with associated minimal modal realization:
    \begin{equation}
    \begin{aligned}
        A_m &= \mathrm{diag}(p_1, \dots,p_n), \quad
        & B_m &= \begin{bmatrix}
            1 & \ldots & 1
        \end{bmatrix}^\top,\\
        C_m &= \begin{bmatrix}
            c_ 1 & \cdots & c_n
        \end{bmatrix}, \quad
        &D_m &= d
    \end{aligned}
    \end{equation}
    Define $T_0 \coloneq\prod_{i=1}^{2(n-1)}E_i$ with 
    \begin{equation}
        E_i = \begin{cases}
            E_{n,i+1}^{x_i}= I_n+x_i\delta_{n,i+1}, &i \leq n-2\\
            E_{i+1,1}^{x_i}=I_n+x_i\delta_{i+1,1}, &i \geq n-1
        \end{cases}
    \end{equation}

    Since the modal realization is a minimal realization and the system is SISO, $c_n\neq 0$. Choose $x$ with
    \begin{equation}
        x_i = \begin{cases}
            \frac{c_{i+1}}{c_n},&\quad i = 1,\cdots,n-2 \\
            -1,&\quad i = n-1,\cdots,2(n-2)
        \end{cases}
    \end{equation}
    Then, $(T_0A_mT_0^{-1},T_0B_m,CT_0^{-1},d)$ satisfies $\Scal$.
\end{proof}

\section{Solving multivariate polynomial equations with Gr\"obner basis} \label{sec:grobner}

Having shown that Problems \ref{pb:complex_realization} and \ref{pb:real_realization} are equivalent to solving systems of multivariate polynomial equations, we now turn to solving them.
Such systems have been studied in the field of algebraic geometry \cite{cox1997ideals,cox1998using,lasserre2008semidefinite}; we leverage these results to introduce algorithms that compute complex- and real-valued solutions, which solve Problems \ref{pb:complex_realization} and \ref{pb:real_realization}, respectively.
Throughout the section, the number of variables is denoted by $k\in\N$ (e.g., $k = n(n-1)$ for \eqref{equ:the_equivalent_multvar_poly_eqns}); the set of polynomials is denoted by $H = \{h_1,\cdots,h_\eta\}$, with $h_i\in\R[x]$ for real solutions and $h_i\in\C[x]$ for complex solutions.

\subsection{Computing complex-valued solutions}

Building on existing theory for the existence of complex-valued solutions \cite{cox1997ideals}, we recast known results to propose Algorithm \ref{alg:comp_solns}.
We first introduce relevant concepts.
Given $H\coloneq\{h_1,\cdots,h_{\eta}\}$, $h_i\in\C[x]$, the \textbf{ideal}\cite[Chapter 1.4]{cox1997ideals} generated by $h_i$ written as $I = \abrac{h_1,\cdots,h_{\eta}}$ is: 
\begin{equation}
\begin{aligned}
    I = \{g\in\C[x]\mid g = \sum_{i=1}^\eta{f_ih_i}, f_i\in\C[x]\}
\end{aligned}
\end{equation}

\tbf{Monomial ordering} is how we order different monomials in polynomials for operations like division. For two distinct monomials $x^\alpha$ and $x^\beta$, it decides if $x^\alpha\succ x^\beta$, so $x^\alpha$ is listed first, or $x^\alpha\prec x^\beta$. We use lexicographic order for our algorithm \cite[Chapter 2.2]{cox1997ideals}. A key property of such an order is that relative ordering $x_1\succ x_2\succ\cdots x_{k}$ holds.

The \tbf{leading term} of a polynomial $f = \sum_{\alpha\in U}c_\alpha x^\alpha\in\C[x]$, denoted by $\LT(f)$, with $U$ being the set of possible degrees of $x$ in $f$, is $c_\beta x^\beta$ where $\forall \alpha\in U\backslash \{\beta\}$, $x^\beta\succ x^\alpha$. 
 
A \tbf{Gr\"obner basis} \cite[Chapter 2.5]{cox1997ideals} of an ideal $I$ of certain monomial ordering is a set of polynomials $G = \{g_1,\cdots,g_\zeta\}$ such that the ideal generated by leading terms of $g_j$, written as $\abrac{\text{LT}(g_1),\cdots,\text{LT}(g_\zeta)}$ equals the ideal generated by the set of leading terms for nonzero polynomials in $I$, written as $\abrac{\text{LT}(I)}$. When $G$ is a Gr\"obner basis, $\abrac{g_1,\cdots,g_\zeta} = I$. A \tbf{reduced Gr\"obner basis} is a special type of Gr\"obner basis with minimal size (i.e., $\abs{G}$). It is unique for a given monomial ordering. 

Gr\"obner basis helps determine the existence of solutions. We first introduce a notation on the solution set. 

For an ideal $I = \abrac{h_1,\cdots,h_\eta}$, one can verify \cite{cox1997ideals} that the complex solution set of $h_i = 0$, $i = 1,\cdots,\eta$ is:
\begin{equation}
    V_\C(I) \coloneq\{x\in\C^k\mid f(x) = 0,\quad \forall f\in I\}
\end{equation}

When $G = \{g_1,\cdots,g_\zeta\}$ is the reduced Gr\"obner basis of $I$ of some monomial order, the existence of complex solutions is given by the following claim \cite[Chapter 4.1]{cox1997ideals}: 
\begin{equation}
    V_\C(I) = \emptyset\text{ if and only if }G=\{1\}
    \label{equ:existence_of_complex_solns}
\end{equation}

There are multiple software tools to compute reduced Gr\"obner bases such as MATLAB, Maple and Magma, so we can directly apply \eqref{equ:existence_of_complex_solns}.

When solutions exist, the elimination theorem helps us find them one variable at a time. We introduce the following version for our algorithm.

\begin{proposition}[{\cite[Chapter 3.1]{cox1997ideals}}]
    \label{thm:elimination_thm}
    When $G = \{g_1,\cdots,g_\zeta\}$ is a Gr\"obner basis of and ideal $I$ with respect to a term order satisfying any monomial involving one of $x_1,\cdots,x_{k-1}$ is greater than all monomials in $\C[x_k]$, then for:
    \begin{equation}
        G^k \coloneq G\cap \Bbbk[x_{k}]  \label{equ:Gll}
    \end{equation}
     Any $x_{k}$ that solves $g_j = 0$, $j = 1,\cdots,\zeta$ solves the polynomial equations generated by $G^l_l$.
\end{proposition}

When $G_k\neq \emptyset$, $x_k$ has a finite number of possible values.
By \eqref{equ:existence_of_complex_solns}, at least one value can be extended to a solution of $h_i = 0$, $i = 1,\cdots,\eta$. 
When $G_k = \emptyset$, $x_k$ is a ``free variable" and we can pick an arbitrary $\tilde{x}_k\in \C$. By the Extension Theorem in \cite[Chapter 3.1]{cox1997ideals}, almost all $\tilde{x}_k$ extends to a solution. We can determine if $\tilde{x}_k$ extends to a solution by computing a new Gr\"obner basis with $x_k$ (the last variable) substituted by $\tilde{x}_k$ (a complex number). This technique is summarized in Algorithm \ref{alg:comp_solns}. 

\begin{algorithm}[h]
\textbf{Input:} $H \coloneq \{h_1,\cdots,h_\eta\}$ \\
\textbf{Output:} Solution $x = (x_1,\cdots,x_k)\in\C^k$ or \tbf{Infeasible}
\caption{Computing complex-valued solutions to a set of multivariate polynomial equations}
\label{alg:comp_solns}
\begin{algorithmic}[1]
    \State Compute reduced Gr\"obner basis $G$ of $H$
    \For{$j = k$ to $1$}
        \State \textbf{if } $G = \{1\}$ \textbf{then return Infeasible}
        \State $G_j \gets G\cap\C[x_j]$
        \State \textbf{if } $G_j = \emptyset$ \textbf{then } Randomly pick $\tilde{x}_j\in\C$
        \State \textbf{else } Find $\tilde{x}_j$ that solves equations generated by $G_j$
        
        \StateWrap{Substitute $\tilde{x}_j$ into $G$ to yield $G'$} 
        \StateWrap{Compute reduced Gr\"obner basis $\tilde{G}$ of $G'$ of order described in Proposition \ref{thm:elimination_thm}.}
        \State \textbf{if} $\Tilde{G} = \{1\}$ \textbf{ then} go to Line 7
        \State $x_j \gets \tilde{x}_j$, $G \gets \tilde{G}$
    \EndFor
    \State \Return $x$
\end{algorithmic}
\end{algorithm}

\subsection{Computing real-valued solutions}

Next, we introduce an approach to solve real-valued solutions with the moment approach. We briefly introduce this approach and our proposed algorithm. More detailed information of the moment approach can be found in \cite{lasserre2008semidefinite}. 

We want to solve for
\begin{equation}
    h_i = 0,\text{ } i = 1,\cdots,\eta,\text{ } h_i\in\R[x_1,\cdots,x_{k}]\coloneq\R[x]  \label{equ:polynomial_eqns_to_solve}
\end{equation}

The set of real-valued solutions to \eqref{equ:polynomial_eqns_to_solve} is:
\begin{equation}
    V_\R(I) \coloneq V_\C(I)\cap \R^{k}
\end{equation}

The moment approach relies on the outcome of a semidefinite programming (SDP) problem designed for polynomial equations with finite solutions. We first introduce some notations involved, then propose the SDP problem, and finally include additional treatments to make this SDP approach applicable to equations with infinite solutions.

\tbf{Monomial orders} for $x^\alpha\coloneq x_1^{\alpha_1}x_2^{\alpha_2}\cdots x_k^{\alpha_k}$ are expressed as $\alpha = (\alpha_1,\alpha_2,\cdots,\alpha_k)\in \N^k$. 

The \tbf{vector form of polynomials} can be given when a monomial order is specified. As an example, for $f(x)= x_1^2+x_1x_2-4\in\R[x_1,x_2]$:
\begin{equation}
\begin{aligned}
    f(x) &= FX \\
    F &= \begin{bmatrix}
        -4 & 0 & 0 & 1 & 1 & 0 & \cdots
    \end{bmatrix}\\
    X &= \begin{bmatrix}
        1 & x_1 & x_2 & x_1^2 & x_1x_2 & x_2^2 & \cdots
    \end{bmatrix}^\top
\end{aligned}
    \label{equ:vector_representation_of_polynomials}
\end{equation}

$F$ is the vector representation. We index it as:
\begin{equation}
    F = \begin{bmatrix}
        f_{0,0} & f_{1,0} & f_{0,1} & f_{2,0} & f_{1,1} & f_{0,2} & \cdots \label{equ:polynomial_indexes}
    \end{bmatrix}
\end{equation}
with $f_{\alpha}$ being the coefficient for $x^\alpha$, $\alpha\in\N^k$. In this example, $k=2$. For larger $k$, we can index $F$ in a similar pattern. When $f(x)$ is a finite-order polynomial, we can truncate $F$ and $X$ appropriately without losing information. For example, in \eqref{equ:vector_representation_of_polynomials}, $f = \sum_\alpha f_\alpha x^\alpha = \sum_{\abs{\alpha}\leq 2}f_\alpha x^\alpha$.

A \tbf{moment matrix} $M(\gamma)$ for a sequence $\gamma_\alpha$, $\alpha\in\N^k$ is an infinite matrix defined as:
\begin{equation}
    \left[M(\gamma)\right]_{\alpha,\beta} \coloneqq \gamma_{\alpha+\beta},\quad \alpha,\beta\in\N^k
\end{equation}

In this matrix, we use $\alpha\in \N^k$ to index row entries and $\beta\in\N^k$ to index column entries. The variable are also indexed by $\N^k$. At row $\alpha$ and column $\beta$, the variable is $\gamma_{\alpha+\beta}$ and for entries with same $\alpha+\beta$ values share the same variable. 

A \tbf{truncated moment matrix} $M_t(\gamma)$ of $M(\gamma)$ is a submatrix of $M(\gamma)$ that keeps $\left[M(\gamma)\right]_{\alpha,\beta}$ with $\sum_{i=1}^n\alpha_i\leq t$ and $\sum_{i=1}^n\beta_i \leq t$. $M_t(\gamma)$ contains a truncated sequence of $\gamma$.


For $f(x) = \sum_{\beta}f_\beta x^\beta$, we define the sequence $f\gamma$ as:
\begin{equation}
    (f\gamma)_\alpha = \sum_{\beta}f_\beta \gamma_{\alpha+\beta},\quad \alpha\in\N^k
\end{equation}

Then, $M(f\gamma)$ is $((f\gamma)_{\alpha+\beta})_{\alpha,\beta\in\N^k}$. 

Now, we can introduce the SDP.

\begin{proposition}[\cite{lasserre2008semidefinite}]
    \label{thm:real_solns_criterion}
    With
    \begin{equation}
        d_j \coloneq \lceil \up{deg}(h_j)/2\rceil,\quad d\coloneq\underset{j = 1,\cdots,\eta}{\max}d_j \label{equ:def_of_d}
    \end{equation}
    \begin{equation}
    \begin{aligned}
        \min_\gamma \quad 1 \quad s.t.\quad& \gamma_0 = 1,\quad M_t(\gamma)\succeq 0, \\ &M_{t-d_j}(h_j\gamma)=0,\quad t\geq d\label{equ:feas_sdp}
    \end{aligned}
    \end{equation}
    \begin{enumerate}
        \item If $\exists t\geq d$ such that no suitable $\gamma$ can be found, then $V_\R(I) = \emptyset$.
        \item When \eqref{equ:feas_sdp} is feasible, we can find $M_t(\gamma)$. Then,  for $1\leq s\leq t$, let $M_s(\gamma)$ be submatrices of $M_t(\gamma)$ s.t. $\sum_{i=1}^n\alpha_i\leq s$, $\sum_{i=1}^n\beta_i \leq s$. If:
        \begin{equation}
            \exists 1\leq s\leq t\text{ s.t. }\rank M_s(\gamma) = \rank M_{s-1}(\gamma) \label{equ:rank_condition}
        \end{equation}
        we can compute solution candidates for \eqref{equ:polynomial_eqns_to_solve}. If $0<\abs{V_\R(I)}<\infty$, there exist a $t$ and corresponding $\Scal$ that returns solutions.
    \end{enumerate}
\end{proposition}

The method for computing solution candidates for a given $\Scal$ is presented directly in \cite{lasserre2008semidefinite} and is quite notationally heavy, so we summarize the general steps here. First, locate a set of linearly independent columns of $M_s(\gamma)$. Then, extract submatrices of $M_s(\gamma)$ based on indices of linearly independent columns. Finally, invert one of these submatrices and perform matrix multiplication; solution candidates of each $x_i$ are eigenvalues of the resulting matrix products.

With Proposition \ref{thm:real_solns_criterion}, we can find solutions by iteratively increasing $d$. Iteration terminates when number of solutions is zero or finite. To address the issue of the non-termination (i.e.,  $\abs{V_\R(I)}=\infty$), we stop iterations at $t = d + r$ for some parameter $r$. When $r$ is sufficiently large, it is more likely that terminating at $t=r$ implies $\abs{V_\R(I)} = \infty$; however, increasing $r$ also increases the time it takes to obtain a solution. This process is summarized in Algorithm \ref{alg:real_solns_base}.

For $\Bbbk \in \{\R,\C\}$, $\abs{V_\Bbbk(I)} = \infty$ implies that there exists some $l\in\{1,\cdots,k\}$ such that we can pick infinite different $x_l$ values to solve \eqref{equ:polynomial_eqns_to_solve}. Since we only need one solution, we can assign an appropriate value to $x_l$. For a set of polynomials with infinite solutions, we can repeat this process until the solution set of the remaining variables is finite.

As $V_\R(I)\subseteq V_\C(I)$, we can find candidates for $l$ by finding $x_l$ which can be freely assigned in $V_\C(I)$. With Proposition \ref{thm:elimination_thm}, this can be achieved by computing Gr\"obner basis $G$ with $l\in\{1,\cdots,k\}$ being the smallest variable and examining if there is a polynomial containing purely $x_l$. When we cannot find one, we can try to assign a value $\tilde{x}_l$ to $x_l$ and use Algorithm \ref{alg:real_solns_base} to validate. This is described in Algorithm \ref{alg:real_solns}. 

Instead of running Algorithm \ref{alg:real_solns}, a more efficient but less rigorous approach to validate $x_l$ is to compute the Gr\"obner basis with $\tilde{x}_l$ substituted. We repeat this process until we get polynomial equations with finite complex solutions. The issue here is that polynomials with nonzero finite complex solutions could have zero real solution. 



\begin{algorithm}[t]
    \textbf{Input:} $H \coloneq \{h_1,\cdots,h_\eta\}$, $r$  \\
    \textbf{Output:} Solution $x = (x_1,\cdots,x_k)\in\R^k$ or \tbf{Infeasible} or \tbf{InfSolnsLikely}\\
    \vspace{-12pt}
    \caption{Computing real-valued solutions to a system of multivariate polynomial equations with finite solutions}
    \label{alg:real_solns_base}
    \begin{algorithmic}[1]
    \For{$t = d$ to $d+r$}
    \StateWrap{Run SDP \eqref{equ:feas_sdp} with polynomials from $H$}
    \State \textbf{if } SDP Infeasible \textbf{then return Infeasible}
    \If{Rank condition \eqref{equ:rank_condition} holds}
    \State Extract solution candidates.
    \If{$\exists$ solution candidate $x$ solving $h_i =  0$}
    \State \Return $x$
    \EndIf
    \EndIf
    \EndFor
    \State \Return \tbf{InfSolnsLikely}
    \end{algorithmic}
\end{algorithm}

\begin{algorithm}[h]
    \caption{Computing real-valued solutions to a system of multivariate polynomial equations}
    \textbf{Input:} $H \coloneq \{h_1,\cdots,h_\eta\}$, $r$ \\
    \textbf{Output:} Solution $x = (x_1,\cdots,x_k)\in\R^k$ or \tbf{Infeasible}\\
    \vspace{-12pt}
    \begin{algorithmic}[1]
    \State Compute reduced Gr\"obner basis $G$ of $H$
    \State \textbf{if } $G = \{1\}$ \textbf{then return Infeasible}
    \State a2Out = Algorithm \ref{alg:real_solns_base}($G$, $r$)
    \While{a2Out = \textbf{InfSolnsLikely}}
    \State $\hat{G}\gets G$.
    \State Find $l$ with Gr\"obner basis, assign $\tilde{x}_l$, plug into $\hat{G}$.
    \State a2Out = Algorithm \ref{alg:real_solns_base}($\hat{G}$,$r$).
    \State \textbf{if} a2Out = $x$ \textbf{then break}
    \State \textbf{else if} a2Out = \textbf{InfSolnsLikely} \textbf{then} $G\gets \hat{G}$
    \EndWhile
    \State \Return{a2Out}
    \end{algorithmic} 
    \label{alg:real_solns}
\end{algorithm}

\textbf{Remark:} In numerical experiments, we observed that in all cases where complex solutions existed, real solutions also existed; in these cases, we can very likely extract a real solution using Algorithm \ref{alg:comp_solns} directly. However, to guarantee retrieval of a real solution, we need to run Algorithm \ref{alg:real_solns}.

\subsection{Finding sparse realizations}

With algorithms to compute solutions to multivariate polynomial equations, we preset Algorithm \ref{alg:find_realizations} to solve our original problems. Set input variable \textbf{complex} to true to solve Problem \ref{pb:complex_realization} and false to solve Problem \ref{pb:real_realization}. To solve Problem \ref{pb:real_realization}, make sure $\mbf{K}(s)$ has real coefficients and poles.
This algorithm nominally iterates over all members $T_p \in \mathcal{P}_p$; however, in practice, setting $T_p = I_p$ often suffices (this is the case for all simulations in the next section). We can also truncate this iteration for improved efficiency.

\begin{algorithm}[h]
    \textbf{Input:} $\mbf{K}(s)$, $\Scal = (\Acal,\Bcal,\Ccal,\Dcal)$, Boolean variable \tbf{complex} \\
    \textbf{Output:} $\Sigma = (A,B,C,D)$ of $\mbf{K}(s)$ satisfying $\Scal$ (real-valued if \tbf{complex} $=$ false) or \tbf{Infeasible}\\
    \vspace{-12pt}
    \caption{Computing sparse realizations}
    \label{alg:find_realizations}
    \begin{algorithmic}[1]
    \State Construct modal realization $\Sigma_m = (A_m,B_m,C_m,D_m)$ 
    \State \textbf{if } $D_m$ does not satisfy $\Scal$ \textbf{then return Infeasible}
    \State Find $\Pcal_p$
    \While{$\exists T_p\in\Pcal_p$ not used}
    \State Construct $\tilde{\Sigma} = (\tA,\tB,\tC,\tD)$ with new $T_p$. 
    \StateWrap{Extract equations by \eqref{equ:the_equivalent_multvar_poly_eqns} from $\tilde{\Sigma}$ based on $\Scal$.}
    \State Run Algorithm \ref{alg:comp_solns} or Algorithm \ref{alg:real_solns} based on \tbf{complex}
    \If{Algorithm returned a solution $x$}
    \StateWrap{Substitute $x$ into $(\tA,\tB,\tC,\tD)$}
    \State \Return $(\tA,\tB,\tC,\tD)$
    \EndIf
    \State \textbf{else return Infeasible}
    \EndWhile
    \end{algorithmic}
    
\end{algorithm}

\section{Numerical Examples} \label{sec:simulations}

We now apply Algorithm \ref{alg:find_realizations} to specific numerical instances of Problems \ref{pb:complex_realization} and \ref{pb:real_realization}. 
All code required to reproduce our results can be found on Github \footnote{https://github.com/YaozDu/Sparse-LTI-system-realization-via-solving-multivariate-polynomial-equations}.

\begin{figure}[h]
    \centering
    \begin{subfigure}[b]{0.48\columnwidth}
        \centering
        \includegraphics[width=\columnwidth]{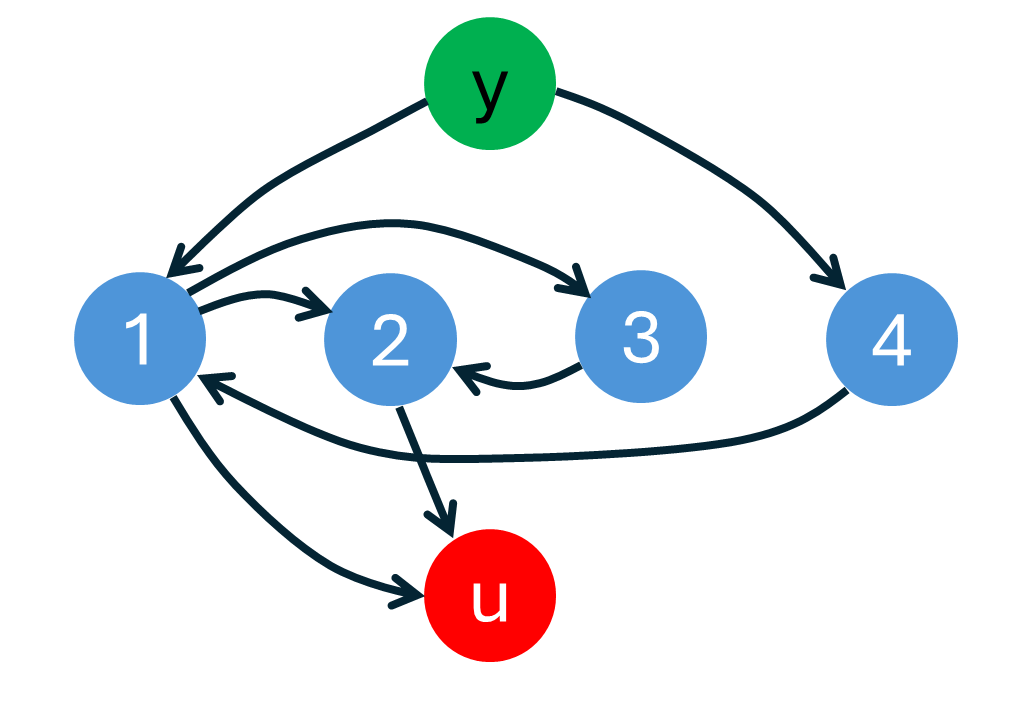}
        \caption{}
        \label{fig:1a}
    \end{subfigure}
    \hfill
    \begin{subfigure}[b]{0.48\columnwidth}
        \centering
        \includegraphics[width=\columnwidth]{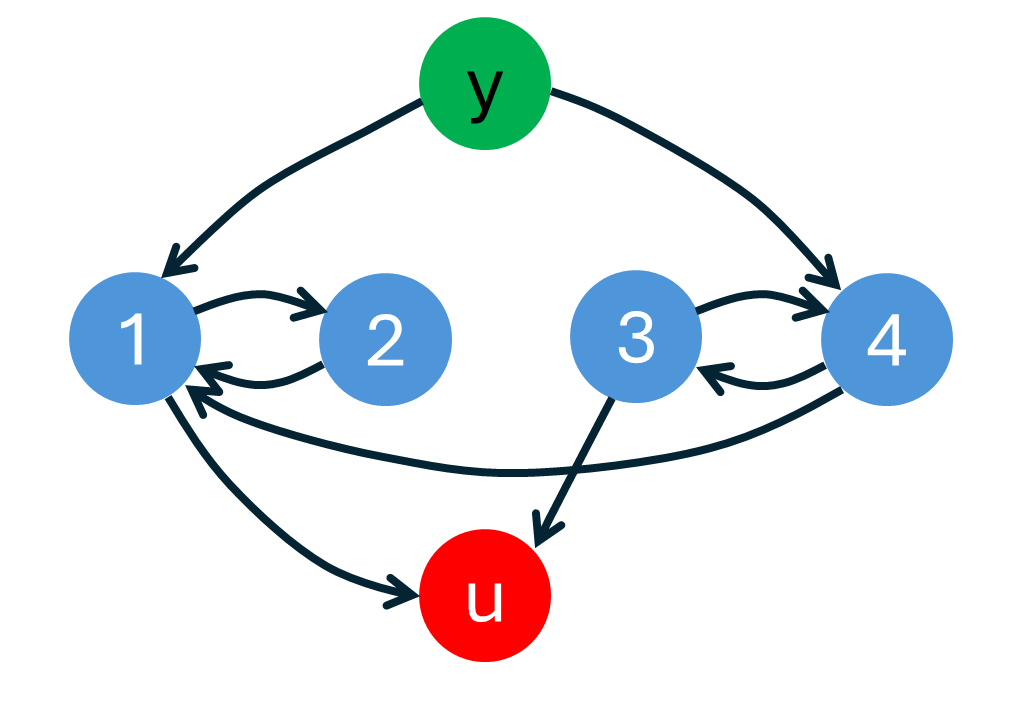}
        \caption{}
        \label{fig:1b}
    \end{subfigure}
    \caption{Controller structures from numerical examples}
    \label{fig:main}
    \vspace{-6pt}
\end{figure}

Consider a linearized two-joint robot arm model:
\begin{equation}
\begin{split}
    F &= \begin{bmatrix}
        0     &    0  &  1    &     0\\
        0     &    0     &    0   & 1\\
        -27.83  & -1.71   &     0    &     0\\
        44.67  &  8.72    &     0   &      0
    \end{bmatrix},
    G = \begin{bmatrix}
        0 \\0\\ 6.93\\-9.54
    \end{bmatrix} \\
    H &= \begin{bmatrix}
        0 & 1 & 0 & 0
    \end{bmatrix}
\end{split}
\end{equation}
with stabilizing LQG feedback controller:
\begin{equation}
    \begin{aligned}
    \dot{z} &= (F+GK-LH)z+Ly,\quad u = Kz \label{equ:LQG} \\
    K &= \begin{bmatrix}
            11.8157  &  8.0213 &   3.8468   & 3.6707
        \end{bmatrix}\\
        L &= \begin{bmatrix}
            -14.8052  & 25.2400 & -13.5943 &  28.2596
        \end{bmatrix}^\top    
    \end{aligned}
\end{equation}


We convert this controller into a transfer function $\mbf{K}(s)$ via $\mbf{K}(s) = K(sI-F-GK+LH)^{-1}L$, and attempt to realize it on various supports. 
We first define the following support, which is depicted in Figure \ref{fig:1a}.
\begin{equation}
\begin{aligned}
    \Acal &= \begin{bmatrix}
        1   &  1    & 0   &  1\\
     1   &  1  &  0   &  0\\
     0   & 0   &  1  &   1\\
     0   &  0  &   1  &   1
    \end{bmatrix},
    \Bcal = \begin{bmatrix}
        1 \\ 0 \\ 0 \\ 1
    \end{bmatrix} \\
    \Ccal &= \begin{bmatrix}
        1 & 0 & 1 & 0
    \end{bmatrix},
    \Dcal = 0
\end{aligned}
\end{equation}
We now wish to solve Problem \ref{pb:real_realization} with this support; this leads to a system of 11 polynomial equations with 12 variables, which has infinite real-valued solutions. 
Algorithm \ref{alg:find_realizations} returns the following realization:
\begin{equation}
    \begin{aligned}
        A &= \begin{bmatrix}
            -17.3543 &  -0.5210  &  0 & -531651.7\\
            -2.5234  &-17.4158  & 0 &  0\\
            0  &  0 &  -17.2105 &-457.3450\\
            0  &  0  & 0.6759  & 18.3806
        \end{bmatrix} \\
    B &= \begin{bmatrix}
        -37.5666 & 0 & 0 & -0.3344
    \end{bmatrix}^\top \\
    C &= \begin{bmatrix}
        -2.1019 & 0 & 2434.126 & 0
    \end{bmatrix}
    \label{equ:implementation_1}
    \end{aligned}
\end{equation}

Next, we consider the support shown in Figure \ref{fig:1b}:
\begin{equation}
    \begin{aligned}
        \Acal &= \begin{bmatrix}
            1 & 0 & 0 & 1 \\ 1 & 1 & 1 & 0 \\ 1 & 0 & 1 & 0 \\ 0 & 0 & 0 & 1
        \end{bmatrix}, 
        \Bcal = \begin{bmatrix}
            1 \\ 0 \\ 0 \\ 1
        \end{bmatrix} \\
        \Ccal &= \begin{bmatrix}
            1 & 1 & 0 & 0
        \end{bmatrix},
        \Dcal = 0
        \label{equ:supp2}
    \end{aligned}
\end{equation}
We now wish to solve Problem \ref{pb:complex_realization} with this support. This leads to a system of 12 polynomial equations with 12 variables. Algorithm \ref{alg:find_realizations} returns the following realization:
\begin{equation}
    \begin{aligned}
        A &= \begin{bmatrix}
        3.3357    &     0   &      0 &-337.8099\\
    3.9752 & -18.5320& -146.8869  &       0\\
    0.2587    &     0  &-16.2381   &      0\\
         0     &    0    &     0  & -2.1656
        \end{bmatrix}, \\
        B &= \begin{bmatrix}
            -712.8518 & 0 & 0 & 0.9144
        \end{bmatrix}^\top \\
        C&= \begin{bmatrix}
            -0.1108 & 1.0888  & 0 & 0
        \end{bmatrix}
    \end{aligned}
    \label{equ:implementation_2}
\end{equation}
As previously remarked, this algorithm returned a real-valued solution despite not being guaranteed to.
Notably, many sparse realizations of this LQG controller exist even though its  original state-space realization \eqref{equ:LQG} is dense.
We consider one final support, identical to the first support except for one element: $\Acal_{12}=0$.
We attempt to solve Problem \ref{pb:complex_realization} via Algorithm \ref{alg:find_realizations}. 
This leads to a system of 12 polynomial equations with 12 variables. This time, the algorithm returns ``Infeasible"; no complex or real-valued realizations exist.

\bibliographystyle{IEEEtran}
\bibliography{References}

@inproceedings{li2025toward,
  title={Toward neuronal implementations of delayed optimal control},
  author={Li, Jing Shuang},
  booktitle={2025 American Ctrl Conf. (ACC)},
  pages={2715--2721},
  year={2025}
}

@article{bullmore2012economy,
  title        = {The economy of brain network organization},
  author       = {Bullmore, Ed and Sporns, Olaf},
  year         = {2012},
  journal      = {Nat. Rev. Neuro.},
  publisher    = {Nature Publishing Group UK London},
  volume       = {13},
  number       = {5},
  pages        = {336--349}
}

@article{anderson2019system,
  title={System level synthesis},
  author={Anderson, James and Doyle, John C and Low, Steven H and Matni, Nikolai},
  journal={Annual Reviews in Control},
  volume={47},
  pages={364--393},
  year={2019},
  publisher={Elsevier}
}

@article{diedrichsen2010coordination,
  title        = {The coordination of movement: optimal feedback control and beyond},
  author       = {Diedrichsen, J{\"o}rn and Shadmehr, Reza and Ivry, Richard B},
  year         = {2010},
  journal      = {Trends in cognitive sciences},
  publisher    = {Elsevier},
  volume       = {14},
  number       = {1},
  pages        = {31--39}
}

@article{nozari2024macroscopic,
  title        = {Macroscopic resting-state brain dynamics are best described by linear models},
  author       = {Nozari, Erfan and Bertolero, Maxwell A and Stiso, Jennifer and Caciagli, Lorenzo and Cornblath, Eli J and He, Xiaosong and Mahadevan, Arun S and Pappas, George J and Bassett, Dani S},
  year         = {2024},
  journal      = {Nature biomedical engineering},
  publisher    = {Nature Publishing Group UK London},
  volume       = {8},
  number       = {1},
  pages        = {68--84}
}

@article{vyas2020computation,
  title        = {Computation Through Neural Population Dynamics},
  author       = {Vyas, Saurabh and Golub, Matthew D. and Sussillo, David and Shenoy, Krishna V.},
  journal      = {Annual Review of Neuroscience},
  volume       = 43,
  number       = 1,
  year         = 2020,
  pages        = {249–275}
}

@article{sabuau2023network,
  title={Network realization functions for optimal distributed control},
  author={Sab{\u{a}}u, {\c{S}}erban and Speril{\u{a}}, Andrei and Oar{\u{a}}, Cristian and Jadbabaie, Ali},
  journal={IEEE Transactions on Automatic Control},
  volume={68},
  number={12},
  pages={8059--8066},
  year={2023},
  publisher={IEEE}
}

@inproceedings{li2020separating,
  title={Separating controller design from closed-loop design: A new perspective on system-level controller synthesis},
  author={Li, Jing Shuang and Ho, Dimitar},
  booktitle={2020 American Ctrl Conf. (ACC)},
  pages={3529--3534},
  year={2020},
  organization={IEEE}
}

@misc{du2024distributed,
  title        = {Distributed Continuous-Time Control via System Level Synthesis},
  author       = {Du, Yaozhi and Li, Jing Shuang},
  year         = 2024,
  note         = {To appear at 2026 American Ctrl Conf. (ACC)}
}

@book{cox1997ideals,
  title={Ideals, varieties, and algorithms},
  author={Cox, David and Little, John and O'Shea, Donal and Sweedler, Moss},
  year={1997},
  publisher={Springer}
}

@book{cox1998using,
  title={Using algebraic geometry},
  author={Cox, David A and Little, John and O’shea, Donal},
  year={1998},
  publisher={Springer}
}

@article{lasserre2008semidefinite,
  title={Semidefinite characterization and computation of zero-dimensional real radical ideals},
  author={Lasserre, Jean Bernard and Laurent, Monique and Rostalski, Philipp},
  journal={Foundations of Comp. Mathematics},
  volume={8},
  number={5},
  pages={607--647},
  year={2008},
  publisher={Springer}
}

@article{de2000minimal,
  title={Minimal state-space realization in linear system theory: an overview},
  author={De Schutter, Bart},
  journal={Journal of computational and applied mathematics},
  volume={121},
  number={1-2},
  pages={331--354},
  year={2000},
  publisher={Elsevier}
}

\end{document}